\documentclass[11pt]{amsart}
\usepackage[english]{babel}
\usepackage{amsfonts,amssymb,amsmath,amsthm}
\usepackage{ctex}
\usepackage{url}
\usepackage{enumerate}
\usepackage{graphicx}
\usepackage{tikz}           
\usepackage{tikz-3dplot} 
\newtheorem{thrm}{Theorem}[section]
\newtheorem{lem}[thrm]{Lemma}
\newtheorem{prop}[thrm]{Proposition}
\newtheorem{cor}[thrm]{Corollary}
\theoremstyle{definition}

\newtheorem{remark}[thrm]{Remark}
\numberwithin{equation}{section}

\begin{document}
\title[running head]{New Periodic Solutions for Newtonian $n$-Body Problems with Dihedral Group Symmetry and  Topological Constraints}

\author{Zhiqiang Wang }
\author{Shiqing Zhang }
\address{Mathematical Department,  Sichuan University, Chengdu, Sichuan 610064, P.R. China}
\begin{abstract}
In this paper, we prove the existence of a family of new non-collision periodic solutions  for the classical Newtonian $n$-body problems.
In our assumption, the $n=2l\geq4$ particles are invariant under the  dihedral rotation group $D_l$ in $\mathbb{R}^3$ such that, at each instant, the $n$ particles form two twisted $l$-regular polygons.
Our approach is variational minimizing method and  we show that the minimizers are collision-free by level estimates and local deformations.

\begin{flushleft}
\text{Keywords:} Periodic Solutions, Newtonian $n$-body Problems, Variational Method, Dihedral Group Symmetry, Topological Constraints.
\end{flushleft}
\end{abstract}
\maketitle
\section{Introduction and   Main Result}

Many authors (for example\cite{chenciner2000remarkable}\cite{ferrario2004existence}\cite{fusco2011platonic}\cite{terracini20072nbody}\cite{1})
 used the variational method to discover many  new periodic solutions of the classical Newtonian $n$-body problems in the last fifteen years. In particular, Chenciner and Montgomery \cite{chenciner2000remarkable} proved the existence of the remarkable figure-8 type periodic solution for planar Newtonian 3-body problems with equal masses. Ferrario and Terracini \cite{ferrario2004existence} simplified and developed Marchal's \cite{marchal2002method}
important works and introduced the rotating circle property, proved that if the motion has certain symmetry under some group action having the rotation circle property, the solution exists and has no collision. Also Fusco et al.\cite{fusco2011platonic} proved the existence and collisionless of a number of new and interesting motions with the invariance of certain platonic polyhedra group action and some topological constraints.

In this paper, we consider a system of $n=2l$ positive masses with their positions $x(t)=(x_1(t),x_2(t),\dots,x_n(t))^T$  moving in the space under Newton's law of gravitation:
\begin{equation}\label{NE}
\mathcal{M}\ddot{x}(t)=\frac{\partial U(x(t))}{\partial x} ,
\end{equation}
where   the potential function  $U(x)=\sum_{i<j}\frac{m_im_j}{|x_i-x_j| }$ and $\mathcal{M}=diag\{m_1,m_2,\dots,m_n\}$.

It is well known that looking for periodic solutions for (\ref{NE}) is equivalent to seeking the critical points of the Lagrange functional $\mathcal{A}:\Lambda\rightarrow \mathbb{R}\cup \{+\infty\}$
\begin{equation*}\begin{aligned}
\mathcal{A}(x(t))&=\int_0^T\mathcal{L}(x)dt=\int_0^T\big(K+U\big)dt\\
&=\int_0^T(\sum_{i=1}^n\frac{1}{2}m_i|\dot{x}_i|^2+\sum_{i<j}\frac{m_im_j}{|x_i-x_j| } )dt
\end{aligned}
\end{equation*}

on the set
\begin{equation*}
\Lambda=\{x(t)\in H^1(\mathbb{R}/T\mathbb{Z},\mathbb{R}^{3n})|x_i(t)\neq x_j(t), \forall i\neq j, \forall t\in\mathbb{R}\},
\end{equation*}
and our approach is based on the following basic lemma:
\begin{lem} (\cite{struwe2008variational})
Let $X$ be a reflexive Banach space, $M \subset X$ is a weakly closed subset, $f :M \rightarrow R $ is weakly lower semi-continuous; if f is coercive, that is, $f(x)\rightarrow +\infty$ as $|x|\rightarrow +\infty$, then f attains its infimum on $M$.
\end{lem}

There are two  difficulties in this approach: one is  the lack of coercivity of the action functional $\mathcal{A}$ on the whole set $\Lambda$;  and the other is  that in  critical points, there might be trajectories with collisions.
To obtain the  coercivity, one can consider the functional $\mathcal{A}$ on some symmetric subspace $\Lambda_G\subset\Lambda$ such that $\mathcal{A}|_{\Lambda_G}$ is coercive. And the following famous lemma proves that  the critical point on $\Lambda_G$ is also a critical point on the whole space $\Lambda$.

\begin{lem}\label{palais}(Palais principle of symmetric criticality \cite{palais1979principle})

Let $G$ be an orthogonal group on a Hilbert space $\Lambda$. Define the fixed
point space: $\Lambda_G = \{x \in \Lambda | g\cdot x = x, \forall g\in G\}$; if $f \in C^1(\Lambda,R)$ and satisfies $f (g\cdot x) = f (x)$
for any $g \in G$ and $x \in\Lambda$, then the critical point of $f$ restricted on $\Lambda_G$ is also a critical point of
$f$ on $\Lambda$.
\end{lem}

Here  is a traditional way to define the group action on loop space $\Lambda$ such that $\mathcal{A}(g\cdot x)=\mathcal{A}(x)$. Let $G$ be a finite group  with three representations $\rho:G\rightarrow O(d)$, $\tau:G\rightarrow O(2)$ and $\sigma: G\rightarrow \sum_n$ such that
$$g\cdot x(t)=(\rho(g)x_{\sigma(g^{-1})(1)}(\tau(g^{-1})t),\dots,\rho(g)x_{\sigma(g^{-1})(n)}(\tau(g^{-1})t)).$$
Where we only consider homomorphisms $\sigma$ with property that $\forall g\in G: (\sigma(g)(i)=j \Rightarrow m_i=m_j)$
and for more detail, we refer the readers to \cite{ferrario2004existence}.
Also we can add some topological constraints on $\Lambda_G$ to get an open cone $\mathcal{K}\subset \Lambda_{G}$ with the property
$$\partial \mathcal{K}\subset \bigtriangleup_{G}= \{u\in \Lambda_{G}:\exists t_c\in\mathbb{R},i\neq j:u_i(t_c)=u_j(t_c) \},$$
and find critical points inside $\mathcal{K}$ (\cite{fusco2011platonic}\cite{montgomery1998braid}\cite{terracini20072nbody}).
By the above arguments, we can distinguish geometrically different solutions, and get the coercivity even if  $\mathcal{A}|_{\Lambda_G}$ is not coercive.
Indeed, if we are able to prove that a minimizer $u_*$ of $\mathcal{A}|_{\overline{\mathcal{K}}}$  exists and for any collision trajectories $u_c\in\partial\mathcal{K}\subset\bigtriangleup_{G}$: $\mathcal{A}(u_*)<\mathcal{A}(u_c)$, then we must have $u_*\in\Lambda_{G}$ and is collision free. Thus by Lemma \ref{palais}, $u_*\in\mathcal{K}$ is a critical point of $\mathcal{A}|_\Lambda$ and therefore a solution of the $n$-body problem.

Suppose the motions are in the space $O-\xi_1\xi_2\xi_3$ and let $e_j$ be the unit vectors of the coordinate axes $\xi_j$ for $j=1,2,3$.
Denote the rotation of angle $\frac{2\pi}{l}$ around $\xi_3$-axis
$R=\begin{pmatrix}
     \cos\frac{2\pi}{l} &-\sin\frac{2\pi}{l} &0\\
     \sin\frac{2\pi}{l} &\cos\frac{2\pi}{l}&0\\
     0 & 0 & 1
     \end{pmatrix}$,
the rotation of angle $\pi$ around $\xi_1$-axis
$S=\begin{pmatrix}
     1 &0 &0\\
     0 &-1&0\\
     0 & 0 & -1
\end{pmatrix}$.
Then $D_l=\langle R,S \rangle$ is the dihedral group (\cite{finite-reflection-groups}) of order $n=2l$ and it is a group of rotations with their rotation axes $\Gamma=\xi_3\bigcup\{L_k\}_{j=0}^{l-1}$, where $L_k$ is the line $\left\{\begin{aligned}\xi_2&=\xi_1\tan\frac{k\pi}{l},\\ \xi_3&=0.
\end{aligned}\right.$

For simplicity, we denote $D_l=\{R_j\}_{j=0}^{n-1}$, where $R_k=R^k$, $R_{l+k}=R^kS$ for $k=0,1,\dots,l-1$, i.e.
$$R_k=\begin{pmatrix}
     \cos\frac{2k\pi}{l} &-\sin\frac{2k\pi}{l} &0\\
     \sin\frac{2k\pi}{l} &\cos\frac{2k\pi}{l}&0\\
     0 & 0 & 1
     \end{pmatrix},
R_{l+k}=\begin{pmatrix}
     \cos\frac{2k\pi}{l} &\sin\frac{2k\pi}{l} &0\\
     \sin\frac{2k\pi}{l} &-\cos\frac{2k\pi}{l}&0\\
     0 & 0 & -1
     \end{pmatrix}.$$
Now we consider $n=2l\geq4$ point particles $u(t)=\big(u_0(t),u_1(t),\dots,u_{n-1}(t)\big)$ with equal masses in space with the following symmetry:
\begin{equation}\label{ca}
u_j(t)=R_ju_0(t),\quad j=0,1,\dots,n-1.
\end{equation}
Under this symmetry the trajectories have the property that, at each instant, the $n$ point particles form a two nested regular $l$-polygons with the same size. By letting the  the masses $m_i=1$ and  under the assumption of the symmetry (\ref{ca}), the action functional can be written as
\begin{equation*}
\mathcal{A}\big(u(t)\big)=\mathcal{A}\big(u_0(t)\big)
=\frac{n}{2}\int_0^T\Big(|\dot{u}_0(t)|^2+\sum_{j=1}^{n-1}\frac{1}{|(R_j-R_0)u_0(t)| } \Big)dt.
\end{equation*}
In \cite{0951-7715-21-6-009}, Ferrario and Portaluri studied central configurations with this dihedral symmetry.
In  Section 2 of  \cite{fusco2011platonic}, Fusco et al. got a new periodic solution in the case of $n=4$ by applying some topological constraint, here our result is a generalization of theirs. Obviously, the trajectories $u(t)$ are uniquely determined by the trajectory $u_0(t)$, and in \cite{fusco2011platonic},  $u_0(t)$ is called  the generating particle of the motion.
%
%
Also we  need some other symmetric  constraints on the loop of $u_0$:
\begin{equation}\label{cb}
\left\{\begin{matrix}
\begin{aligned}
u_0(t)&=\hat{R}_0u_0(-t),\\
u_0(t)&=\hat{R}_su_0(\frac{T}{h}-t).
\end{aligned}
\end{matrix}\right.
\end{equation}
Where $s,h\leq l$ are some positive integers and
 $\hat{R}_k=\begin{pmatrix}
     \cos\frac{2k\pi}{l} &\sin\frac{2k\pi}{l} &0\\
     \sin\frac{2k\pi}{l} &-\cos\frac{2k\pi}{l}&0\\
     0 & 0 & 1
     \end{pmatrix}$
 are reflections along the plane $P_k:\xi_2=\xi_1\tan\frac{k\pi}{l}$, $k=0,1,\dots,l-1$.\\

We notice that (\ref{cb}) implies $u_0(t)=\hat{R}_s\hat{R}_0u_0(t-\frac{T}{h})=R_su_0(t-\frac{T}{h})$. Since $u_0(t)$ is a $T$-periodic solution i.e. $u_0(t)=u_0(t+T)$, we must have $R_s^h=R_0=id$ which implies
\begin{equation}\label{hmod}
sh\equiv l\mod l.
\end{equation}
Moreover if $T$ is the minimal positive period, then $h$ is the minimal positive integer satisfying (\ref{hmod}) and of course $h\mid l$. Also applying (\ref{ca}) we see that $u_0(t)=u_s(t-\frac{T}{h})$ and
$$\begin{aligned}
u_l(t)&=Su_0(t)=SR_su_0(t-\frac{T}{h})=R_{l-s}Su_0(t-\frac{T}{h})\\
      &=R_{2l-s}u_0(t-\frac{T}{h})=u_{2l-s}(t-\frac{T}{h}).
\end{aligned}$$
If we denote $\bar{j}$ to be the least nonnegative residue of $j$ modulo $l$, i.e. $\bar{j}\equiv j \mod l$ and $0\leq \bar{j}<l$, we have
\begin{equation*}
\left\{\begin{matrix}
\begin{aligned}
u_0(t)&=u_{\bar{s}}(t-\frac{T}{h})=u_{\overline{2s}}(t-\frac{2T}{h})=\dots=u_{\overline{(h-1)s}}(t-\frac{(h-1)T}{h}),\\
u_1(t)&=u_{\overline{s+1}}(t-\frac{T}{h})=u_{\overline{2s+1}}(t-\frac{2T}{h})=\dots=u_{\overline{(h-1)s+1}}(t-\frac{(h-1)T}{h}),\\
\ldots\\
u_{l/h-1}(t)&=u_{\overline{s+l/h-1}}(t-\frac{T}{h})=u_{\overline{2s+l/h-1}}(t-\frac{2T}{h})=\dots=u_{\overline{(h-1)s+l/h-1}}(t-\frac{(h-1)T}{h}),\\
u_l(t)&=u_{2l-\bar{s}}(t-\frac{T}{h})=u_{2l-\overline{2s}}(t-\frac{2T}{h})=\dots=u_{2l-\overline{(h-1)s}}(t-\frac{(h-1)T}{h}),\\
u_{2l-1}(t)&=u_{2l-(\overline{s+1})}(t-\frac{T}{h})=u_{2l-(\overline{2s+1})}(t-\frac{2T}{h})=\dots=u_{2l-(\overline{(h-1)s+1})}(t-\frac{(h-1)T}{h}),\\
\ldots\\
u_{2l-l/h+1}(t)&=u_{2l-(\overline{s+l/h-1})}(t-\frac{T}{h})=u_{2l-(\overline{2s+l/h-1})}(t-\frac{2T}{h})=\dots\\
&=u_{2l-(\overline{(h-1)s+l/h-1})}(t-\frac{(h-1)T}{h}).
\end{aligned}
\end{matrix}\right.
\end{equation*}
That is to say the $n=2l$ particles' motion are  composed of $2l/h$ choreography trajectories.
For example when $l=12,s=9$, then we have $h=4$ and
\begin{equation*}
\left\{\begin{matrix}
\begin{aligned}
u_0(t)&=u_9(t-\frac{T}{4})=u_6(t-\frac{T}{2})=u_3(t-\frac{3T}{4}),\\
u_1(t)&=u_{10}(t-\frac{T}{4})=u_7(t-\frac{T}{2})=u_4(t-\frac{3T}{4}),\\
u_2(t)&=u_{11}(t-\frac{T}{4})=u_8(t-\frac{T}{2})=u_5(t-\frac{3T}{4}),\\
u_{12}(t)&=u_{15}(t-\frac{T}{4})=u_{18}(t-\frac{T}{2})=u_{21}(t-\frac{3T}{4}),\\
u_{23}(t)&=u_{14}(t-\frac{T}{4})=u_{17}(t-\frac{T}{2})=u_{20}(t-\frac{3T}{4}),\\
u_{22}(t)&=u_{13}(t-\frac{T}{4})=u_{16}(t-\frac{T}{2})=u_{19}(t-\frac{3T}{4}).
\end{aligned}
\end{matrix}\right.
\end{equation*}

\begin{remark}\label{s}
We notice that the angle between the plane $P_s$ and $P_0$ is the same with the angle between $P_{l-s}$ and $P_0$, and $hs\equiv l\mod l$ if and only if $(l-s)h\equiv l\mod l$ which means $l-s$ and $s$ implies the same $h$. So in (\ref{cb}), the two symmetric constraints are  the same in geometry, and we can only consider the integer $1\leq s \leq \frac{l}{2}=n/4$.
%
\end{remark}
\begin{remark}\label{foundamental}
By the condition (\ref{cb}), we see that $\mathbb{I}=[0,\frac{T}{2h}]$ is a fundamental domain of dihedral type for the trajectories (see \cite{ferrario2004existence}) for more details), which implies that the motion of the particles on the whole period $[0,T]$ is determined by their motion on $\mathbb{I}=[0,\frac{T}{2h}]$ through the symmetric conditions. And from (\ref{ca}) we see that $u_0$ is the generating particle, so in Section \ref{collision} we can only consider the motion of the generating particle $u_0$ in the interval  $\mathbb{I}=[0,\frac{T}{2h}]$.
\end{remark}

Let $G_s=\langle R, S, \hat{R}_s, \hat{R}_0 \rangle$ with the following representations:
$$\begin{aligned}\rho(R)&=R,\tau(R)t=t,\sigma(R)=(0,1,\dots,l-1)(l,l+1,\dots,n-1),\\
\rho(S)&=S,\tau(S)t=t,\sigma(S)=(0,l)\prod_{k=1}^{l-1}(l-k,l+k),\\
\rho(\hat{R}_s)&=\hat{R}_s,\tau(\hat{R}_s)t=\frac{T}{h}-t,\sigma(\hat{R}_s)=id,\\
\rho(\hat{R}_0)&=\hat{R}_0,\tau(\hat{R}_0)t=-t,\sigma(\hat{R}_0)=id,\end{aligned}$$
and set
$$\Lambda_s=\{u(t)\in\Lambda: u(t)\quad satisfy \quad (\ref{ca})(\ref{cb})\}.$$
It is easy to check that $\Lambda_s=\Lambda_{G_s}$, i.e. looking for trajectories with properties (\ref{ca})(\ref{cb}) is equivalent to seeking for critical point of $\mathcal{A}$ on $\Lambda_{G_s}$.

In \cite{ferrario2004existence}, Ferrario and Terracini proved that $\mathcal{A}|_{\Lambda_G}$ is coercive if and only if $\mathcal{X}_G \triangleq\{x\in\mathcal{X}|g\cdot x=x,\forall g\in G\}=0$ where $\mathcal{X}$ is the configuration space of the particles. Obviously, in our assumption, $\mathcal{A}|_{\Lambda_s}$ is not coercive since $(e_3,\dots,e_3,-e_3,\dots,-e_3)\in\mathcal{X}_G$ where $e_3=(0,0,1)$.
So motivated by \cite{fusco2011platonic}, we add some topological condition on $\Lambda_{G_s}$, to get the open cone $\mathcal{K}_s$ described in the previous. From (\ref{cb}) we see that  $u_0(0)\in P_0$ and $u_0(\frac{T}{2h})\in P_s$, so we let
\begin{equation}\label{cone}
\mathcal{K}_s=\{u(t)\in\Lambda_s:u_{0}(0)\in P_0^-,  u_{0}(\frac{T}{2h})\in P_s^+\}
\end{equation}
where we have set $P_0^-=\{p\in P_0: p\cdot e_3<0\}$ and $P_s^+=\{p\in P_s: p\cdot e_3>0\}$.
Now we state our main  theorem:
\begin{thrm} \label{main}
For $n=2l\geq4$ and every integer $s\leq\max\{1,(\frac{n-1}{n})^{3/2}\frac{\pi}{2^{3/2}}\frac{n}{\log n+\gamma}-1\}$ where $\gamma\approx0.57721566490153286$ is the Euler-Mascheroni constant,
there exists a $T$-periodic solution $u_*\in\mathcal{K}_s$   of the classical $n$-body problem.
\end{thrm}
\begin{cor}
For $n=2l\geq4$ and $s=1$, then $h=l$ and there exists a $T$-periodic solution $u_*\in\mathcal{K}_1$  of the classical Newtonian $n$-body problem.
\end{cor}
The case for $n=4$ was discussed in Section 2 of \cite{fusco2011platonic}, here we generalize their result. By  Remark \ref{s}, $s$ can be chosen larger than $1$ when $n\geq8$.
If we let
$f(n)=(\frac{n-1}{n})^{3/2}\frac{\pi}{2^{3/2}}\frac{n}{\log n+\gamma}-1$, we see that $f(n)$ is monotone increasing and we have $f(4)\approx0.4697$, $f(6)\approx1.1400$, $f(8)\approx1.7376$, $f(10)\approx2.2931$, $f(14)\approx3.3262$, $f(26)\approx6.0995$, etc. That is to say, for $n\geq10$, we can choose $s=2$ such that there exists a $T$-periodic solution $u_*\in\mathcal{K}_2$  of the classical $n$-body problem.
\begin{remark}\label{8}
Actually, the value of integer $s$ depends on the estimate of excluding total collisions. By doing more explicit computation we can prove that, for $n=8$, there exists a $T$-periodic solution $u_*\in\mathcal{K}_2$  of the classical $8$-body problem.
\end{remark}

\section{Coercivity}

\begin{prop}
$\mathcal{A}|_{\mathcal{K}_s}$ is coercive and $\partial \mathcal{K}_s\subset \bigtriangleup_{G_s}$.
\end{prop}
\begin{proof}
From (\ref{ca}) we see that there is a collision if and only if there exists some $t_c\in\mathbb{R}$ such that $u_0(t_c)\in\Gamma$, where $\Gamma$ are the rotating axes of the dihedral group $D_l$. So it is obvious that $\partial \mathcal{K}_s\subset \bigtriangleup_{G_s}$ and in the following we will prove the coercivity.

Applied Newton-Leibniz Formula and H$\ddot{o}$lder Inequality, we have
$$\begin{aligned}
|u_0(T/2h)-u_0(0)|&=|\int_0^\frac{T}{2h}\dot{u}_0(t)dt|\leq \int_0^\frac{T}{2h}|\dot{u}_0(t)|dt \\
               &\leq  (\frac{T}{2h})^\frac{1}{2} (\int_0^\frac{T}{2h}|\dot{u}_0(t)|^2dt)^\frac{1}{2}
\end{aligned}.$$
Since $u_{0}(0)\in P_0^-,  u_{0}(\frac{T}{2h})\in P_s^+$, we must have
$$u_{0}(0)\cdot u_{0}(\frac{T}{2h})\leq |u_{0}(0)| |u_{0}(\frac{T}{2h})|\cos\frac{s\pi}{l},$$
which induces that
$$\begin{aligned}
||\dot{u}_0(t)||_{L^2}^2&\geq C_1 |u_0(\frac{T}{2h})-u_0(0)|^2\\
                        &\geq C_1\Big[|u_0(\frac{T}{2h})|^2+|u_0(0)|^2-2|u_{0}(0)| |u_{0}(\frac{T}{2h})|\cos\frac{s\pi}{l}\Big]\\
                        &= C_1\Big[|u_0(0)|\cos\frac{s\pi}{l}- |u_{0}(\frac{T}{2h})|\Big]^2+ C_1|u_0(0)|^2\sin^2\frac{s\pi}{l}\\
                        &\geq C |u_0(0)|^2,
\end{aligned}$$
which implies that $\|\dot{u}_0\|_{L^2}$ is an equivalent norm of $H^1=W^{1,2}$ and the coercivity for the functional $\mathcal{A}$ follows.
\end{proof}

\section{Estimate on Collisions and the Proof of   Theorem \ref{main}  }\label{collision}
In this section, we show that the minimizer $u_*\in\overline{\mathcal{K}}_s$ is free of collisions. Firstly we exclude total collisions and in Section \ref{secpc} we discuss the partial collisions.

\subsection{Total Collision}
Our way to show that there is no total collision is based on level estimate. Firstly  assuming that   a total collision happens in $u_*\in\overline{\mathcal{K}}_s$, we show that there is a lower bound of the action functional $\mathcal{B}\leq \mathcal{A}(u_*) $, and we construct a test loop $\tilde{u}\in\mathcal{K}_s$ without collisions such that $\mathcal{A}(\tilde{u})< \mathcal{B}$. Then we have $\mathcal{A}(\tilde{u})<\mathcal{A}(u_*)$ which contradicts with that $u_*$ is a minimizer. Thus the minimizers $u_*\in\overline{\mathcal{K}}_s$ is free of total collisions.

\begin{lem}\label{gordan}(Gordon's Theorem \cite{gordon1977minimizing})
Let $x\in W^{1,2}([t_1,t_2],R^d)$ and $x(t_1)=x(t_2)=0$. Then for any $a>0$, we have
$$\int_{t_1}^{t_2}(\frac{1}{2}|\dot{x}|^2+\frac{a}{|x|})dt\geq \frac{3}{2}(2\pi)^\frac{2}{3}a^\frac{2}{3}(t_2-t_1)^\frac{1}{3}.$$
\end{lem}

\begin{prop}\label{lbe} (Lower bound estimates for $\mathcal{A}$ with total collisions)\\
Assume that $u\in\overline{\mathcal{K}}_s$ has a total collision. Then
$$\mathcal{A}(u)\geq
\frac{3(n-1)}{4}(2h)^\frac{2}{3}\pi^\frac{2}{3}n^\frac{2}{3}T^\frac{1}{3}\triangleq\mathcal{B}.$$
\end{prop}
\begin{proof}
Since the center of mass is at origin, the functional $\mathcal{A}$ can be written as
\begin{equation}
\mathcal{A}(u)=\frac{1}{M}\sum_{i<j}m_im_j\int_0^T[\frac{1}{2}|\dot u_i-\dot u_j|^2+\frac{M}{|u_i-u_j|}]dt,
\end{equation}
where the total mass $M=\sum_{i=1}^Nm_i$. This formulation came from (\cite{16196083}\cite{16136657}) and has been widely used to obtain the lower bound estimate of  collision paths (\cite{chen2008existence},\cite{Chenciner2003ICM},\cite{1},\cite{2}etc). In our assumption, we have $m_1=\dots=m_n=1$ and $M=n$. Moreover, if $u\in\overline{\mathcal{K}}_s$ has a total collision, then they collide at least $h$ times in the interval $[0,T)$. Applying Lemma \ref{gordan}, we have
$$\begin{aligned}
\mathcal{A}(u)
&=\frac{1}{n}\sum_{i<j}\int_0^T[\frac{1}{2}|\dot u_i-\dot u_j|^2+\frac{n}{|u_i-u_j|}]dt\\
&\geq\frac{1}{n}\frac{n(n-1)}{2}\frac{3}{2}(2\pi)^\frac{2}{3}n^\frac{2}{3}(\frac{T}{h})^\frac{1}{3}\times h\\
&=\frac{3(n-1)}{4}(2h)^\frac{2}{3}\pi^\frac{2}{3}n^\frac{2}{3}T^\frac{1}{3}\triangleq\mathcal{B}.
\end{aligned}$$
\end{proof}

 Next we construct test loops $\tilde{u}\in\mathcal{K}_s$ such that $\mathcal{A}(\tilde{u})<\mathcal{B}$ in two different ways. The idea of the first one is from Fusco et al.\cite{fusco2011platonic}, but it holds only for $s=1$. The second one holds for $s<(\frac{n-1}{n})^{3/2}\frac{\pi}{2^{3/2}}\frac{n}{\log n+\gamma}-1$ and it needs some more explicit analysis on the potential $U$.

\begin{prop}\label{ubes1} (Upper bound estimate for $s=1$)\\
When $s=1$, there exists $\tilde{u}\in\mathcal{K}_1$ such that
$$\mathcal{A}(\tilde{u})\leq \frac{3}{4}(2h^2)^\frac{1}{3}n(n-1)^\frac{2}{3}\pi^\frac{2}{3}T^\frac{1}{3}.$$
\end{prop}
\begin{proof}
We construct the  test loop $\tilde{u}$ similar to that in Proposition 5.3 in \cite{fusco2011platonic}. Assume the generating particle $\tilde{u}_0$ moves with constant speed on a  curve which is the union of two quarters of circumferences $C_1$, $C_2$ of radius $r\tan\frac{\pi}{n}$. $C_1$ has the center $(r,0,0)$ and lies on the plane $\xi_1=r$.  $C_2$ has the center $(r\cos\frac{\pi}{l},r\sin\frac{\pi}{l},0)$
and lies on the plane ${\xi_1}{\cos\frac{\pi}{l}}+{\xi_2}{\sin\frac{\pi}{l}}=r$, see figure 1.
\begin{figure}
\tdplotsetmaincoords{60}{110}
\begin{tikzpicture}[scale=3,tdplot_main_coords]
\draw[thick,<-] (-1,0,0) node[left]{$\xi_2$} -- (0.3,0,0);
\draw[thick,->] (0,-0.1,0) -- (0,1.5,0) node[anchor=north west]{$\xi_1$};
\draw[thick,->] (0,0,-0.2) -- (0,0,1) node[anchor=south]{$\xi_3$};
\draw (0,0,0) -- (-0.866,0.5,0) ;
\draw [dashed](0,0,0) -- ((-0.577,1,0) ;
\coordinate (O) at (0,0,0);
\tdplotsetcoord{P}{.8}{50}{70}
\tdplotsetthetaplanecoords{60}
\tdplotdrawarc[tdplot_rotated_coords,very thin,->]{(0,0,1)}{0.577}{0}{90}{}{}
\tdplotsetthetaplanecoords{0}
\tdplotdrawarc[tdplot_rotated_coords,very thin,->]{(0,0,1)}{0.577}{180}{270}{}{}
\tdplotdrawarc[tdplot_rotated_coords,dashed]{(0,0,1)}{0.577}{0}{180}{}{}
\tdplotdrawarc[tdplot_rotated_coords,dashed]{(0,0,1)}{0.577}{270}{360}{}{}
\draw [dashed](0,1,0) -- (-0.577,1,0)node[right] {$u_0(\frac{T}{4h})$};
\draw [dashed](0,1,0) -- (0,1,-0.577)node[right] {$u_0(0)$};
\draw [dashed](-0.577,1,0) -- (-0.866,0.5,0)  ;
\draw [dashed](-0.866,0.5,0.577)node[right] {$u_0(\frac{T}{2h})$} -- (-0.866,0.5,0)  ;
\filldraw [gray] (0,1,-0.577) circle (0.7pt)
                 (-0.577,1,0) circle (0.7pt)
                 (-0.866,0.5,0.577) circle (0.7pt);
\end{tikzpicture}
\caption{}
\end{figure}
So the constant speed of generating particle $\dot{\tilde{u}}_0=\frac{\pi r\tan\frac{\pi}{n}}{T/2h}$,  and the kinetic energy
$$K(\dot{\tilde{u}})=\frac{n}{2}(\frac{2h\pi r\tan\frac{\pi}{n}}{T})^2.$$

From the definition of $\tilde{u}$, we see that $|\tilde{u}_i-\tilde{u}_j|\geq 2r\tan\frac{\pi}{n}$ for all $i\neq j$, which implies $$V(\tilde{u})\leq\frac{n(n-1)}{2}\frac{1}{2r\tan\frac{\pi}{n}},$$
therefore we have
$$\mathcal{A}(\tilde{u})=\int_0^TK(\dot{\tilde{u}})+V(\tilde{u})dt\leq\frac{2h^2n\pi^2r^2\tan^2\frac{\pi}{n}}{T}+\frac{n(n-1)T}{4r\tan\frac{\pi}{n}},$$
and the conclusion follows if we choose $r=\frac{(n-1)^{1/3}T^{2/3}}{(16h^2)^\frac{1}{3}\tan\frac{\pi}{n}\pi^{2/3}}$.
\end{proof}
From Proposition \ref{lbe} and \ref{ubes1}, we see that
$\frac{\mathcal{A}(\tilde{u})}{\mathcal{B}}\leq(\frac{n}{2(n-1)})^{1/3}<1$, i.e. $\mathcal{A}(\tilde{u})<\mathcal{B}$.
Next we consider the situation    $s\geq2$, since $s\leq\frac{n}{4}$, in the following we suppose $n=2l\geq8$.

\begin{prop}\label{upperbound}
(Upper bound estimate for the general  $s$)\\
When $n=2l\geq8$, there is a test loop $\tilde{u}\in\mathcal{K}_s$ such that $\mathcal{A}(\tilde{u})<\mathcal{B}$  for $s+1<(\frac{n-1}{n})^{3/2}\frac{\pi}{2^{3/2}}\frac{n}{\log n+\gamma}$.
\end{prop}
\begin{proof}
Assume that the particle $\tilde{u}_0$ moves with constant speed on the sphere $|\tilde{u}_0|=a$, with the radius $a$ to be determined later. More precisely, suppose $\tilde{u}_0(\varphi(t),\theta(t))=a(\cos\varphi e^{\theta(t)\sqrt{-1}},\sin\varphi(t))$ where
$$\left\{\begin{aligned}
(\varphi(t),\theta(t))&=(-\frac{\pi}{n},\omega t), & t\in[0,\frac{T}{4h(s+1)}] \\
(\varphi(t),\theta(t))&=(\omega t-\frac{2\pi}{n},\frac{\pi}{n}), & t\in(\frac{T}{4h(s+1)},\frac{3T}{4h(s+1)}]\\
(\varphi(t),\theta(t))&=(\frac{\pi}{n},\omega t-\frac{2\pi}{n}), & t\in(\frac{3T}{4h(s+1)},\frac{T}{2h}]\\
\end{aligned}\right.$$
where $\omega=\frac{(s+1)\pi/l}{T/2h}=\frac{2h(s+1)\pi}{lT}$, see Figure 2.
\begin{figure}
\tdplotsetmaincoords{60}{110}
\begin{tikzpicture}[scale=3,tdplot_main_coords]
\draw[thick,<-] (-1.5,0,0) node[right]{$\xi_2$} -- (1.5,0,0);
\draw[thick,->] (0,-1.5,0) -- (0,1.5,0) node[anchor=north west]{$\xi_1$};
\draw[thick,->] (0,0,-1) -- (0,0,1) node[anchor=south]{$\xi_3$};

\tdplotdrawarc[very thin,dashed]{(0,0,0.3)}{0.7}{0}{360}{}{}
\tdplotdrawarc[very thin,dashed]{(0,0,-0.3)}{0.7}{0}{360}{}{}
\tdplotdrawarc{(0,0,-0.3)}{0.7}{90}{120}{}{}
\tdplotdrawarc{(0,0,0.3)}{0.7}{120}{150}{}{}
\draw [very thin](0,0,-0.3) -- (0,0.7,-0.3)  ;
\draw [very thin](0,0,-0.3) -- (-0.35,0.606,-0.3)  ;
\draw [very thin](0,0,0.3) -- (-0.35,0.606,0.3)  ;
\draw [very thin](0,0,0.3) -- (-0.606,0.35,0.3)  ;
\draw [very thin,dashed](0,0,0) -- (-0.35,0.606,-0.3)  ;
\draw [very thin,dashed](0,0,0) -- (-0.35,0.606,0.3)  ;
\filldraw [gray] (0,0.7,-0.3) circle (0.5pt) node[right]{$u_0(0)$}
                 (-0.35,0.606,-0.3) circle (0.5pt)
                 (-0.35,0.606,0.3) circle (0.5pt)
                 (-0.606,0.35,0.3) circle (0.5pt)node[right]{$u_0(T/2h)$};
\draw (-0.35,0.606,-0.3) .. controls (-0.47,0.72,-0.1) and (-0.47,0.72,0.1)  .. (-0.35,0.606,0.3);

\end{tikzpicture}
\caption{}
\end{figure}

First we claim that, for every $t\in[0,\frac{T}{2h}]$, the potential
\begin{equation}\label{potential estimate}
U(\tilde{u}(t))<\frac{n^2}{2a\pi}(\log n+\gamma).
\end{equation}

And the kinetic energy
$$K(\dot{\tilde{u}})<\frac{n}{2}(\omega a)^2=\frac{2h^2na^2(s+1)^2\pi^2}{l^2T^2},$$
so the functional
$$\mathcal{A}(\tilde{u})=\int_0^T[K+U ]dt<\int_0^T[\frac{2h^2na^2(s+1)^2\pi^2}{l^2T^2}+\frac{n^2}{2a\pi}(\log n+\gamma)] dt.$$
we choose $a=\frac{n^{1/3}T^{2/3}(\log n+\gamma)^{1/3}}{2\pi (s+1)^\frac{2}{3}}(\frac{l}{h})^{2/3}$, then
$$
\mathcal{A}(\tilde{u})<\frac{3}{2}(\frac{h}{l})^\frac{2}{3}n^\frac{5}{3}(\log n+\gamma)^{2/3}(s+1)^{2/3}T^{1/3},$$
$$\frac{\mathcal{A}(\tilde{u})}{\mathcal{B} }<\frac{n}{n-1}\frac{2}{\pi^{2/3}}(\frac{\log n+\gamma}{n})^{2/3}(s+1)^{2/3}<1,$$
and the conclusion follows.

\end{proof}

Now we prove the estimate (\ref{potential estimate}), first we need some more explicit estimate on the potential $U$. We notice that, at every instant, the $n$ particles form a twisted regular $l$-polygons by the symmetric conditions.
Let $u_0=(a\cos\varphi e^{\theta\sqrt{-1}},a\sin\varphi)$, then the potential function can be written as
\begin{equation}\label{potential1}
\begin{aligned}
U&=U(u_0)=U(\varphi,\theta)=\frac{n}{2}\sum_{j=1}^{n-1}|(R_j-R_0)u_0|^{-1}\\
 &=\frac{n}{2a}\Big\{\sum_{j=1}^{l-1}(\cos\varphi)^{-1}|1-\xi_l^j|^{-1}+
 \sum_{j=1}^l[\cos^2\varphi|1-e^{-2\theta\sqrt{-1}}\xi_l^j|^2+4\sin^2\varphi]^{-1/2}\Big\}\\
&=\frac{n}{4a\cos\varphi}\big\{2C_l+\sum_{j=1}^l
[\sin^2(\frac{j\pi}{l}-\theta)+\tan^2\varphi]^{-\frac{1}{2}}\big\},
\end{aligned}
\end{equation}
where $C_l=\sum_{j=1}^{l-1}|1-\xi_l^j|^{-1}=\frac{1}{2}\sum_{j=1}^{l-1}\csc\frac{j\pi}{l}$ and $\xi_l=e^{\frac{2\pi}{l}\sqrt{-1}}$.

Also we notice that
\begin{equation*}
\begin{aligned}
&\quad\cos^2\varphi|1-e^{-2\theta\sqrt{-1}}\xi_l^j|^2+4\sin^2\varphi \\
&=\cos^2\varphi(2-e^{-2\theta\sqrt{-1}}\xi_l^j-e^{2\theta\sqrt{-1}}\xi_l^{-j})+4\sin^2\varphi \\
&=2+2\sin^2\varphi-\cos^2\varphi(e^{-2\theta\sqrt{-1}}\xi_l^j+e^{2\theta\sqrt{-1}}\xi_l^{-j})\\
&=(1+\sin\varphi)^2[\frac{2+2\sin^2\varphi}{(1+\sin\varphi)^2}-
 \frac{1-\sin^2\varphi}{(1+\sin\varphi)^2}(e^{-2\theta\sqrt{-1}}\xi_l^j+e^{2\theta\sqrt{-1}}\xi_l^{-j})]\\
&=(1+\sin\varphi)^2\Big[1+(\frac{1-\sin\varphi}{1+\sin\varphi})^2-
 \frac{1-\sin\varphi}{1+\sin\varphi}(e^{-2\theta\sqrt{-1}}\xi_l^j+e^{2\theta\sqrt{-1}}\xi_l^{-j})\Big]\\
&=(1+\sin\varphi)^2|1-\frac{1-\sin\varphi}{1+\sin\varphi}e^{-2\theta\sqrt{-1}}\xi_l^j|^2.
\end{aligned}
\end{equation*}
That is
\begin{equation*}
U(\varphi,\theta)=\frac{n}{2a\cos\varphi}
\Big(C_l+\frac{\cos\varphi}{1+\sin\varphi}
\sum_{j=1}^l|1-\frac{1-\sin\varphi}{1+\sin\varphi}e^{-2\theta\sqrt{-1}}\xi_l^j|^{-1}\Big),
\end{equation*}
let $r=\frac{1-\sin\varphi}{1+\sin\varphi}$ and $\xi=e^{-2\theta\sqrt{-1}}$, then
\begin{equation}\label{potential2}
U(r,\xi)=\frac{n(1+r)}{4a\sqrt{r}}
(C_l+\sqrt{r}\sum_{j=1}^l|1-r\xi\xi_l^j|^{-1}),
\end{equation}
since there is an  integral representation of $\sum_{j=1}^l|1-r\xi\xi_l^j|^{-1}$,  we can state
\begin{lem}
For $r\in(0,1)$ and $\xi=e^{-2\theta\sqrt{-1}}$ the potential $U(r,\xi)$ can be written as
\begin{equation}\label{potential3}
U(r,\xi)=\frac{n(1+r)}{4a\sqrt{r}}
\Big[C_l+\frac{l\sqrt{r}}{\pi}\int_0^1\frac{(1-t)^{-1/2}t^{-1/2}}{(1-tr^2)^{1/2}}
\frac{1-(tr)^{2l}}{|1-(tr\xi)^l|^{2}}dt\Big].
\end{equation}
\end{lem}
\begin{proof}
In \cite{0951-7715-21-6-009}, Ferrario and Portaluri gave  a general integral representation of $\sum_{j=1}^l|1-r\xi\xi_l^j|^{-\alpha}$ with $0<\alpha<2$. Since this lemma is rather important to our analysis, we prove it here again in the case of $\alpha=1$.
$$\begin{aligned}
|1-r\xi\xi_l^j|^{-1}
&=(1-r\xi\xi_l^j)^{-1/2}(1-r\xi^{-1}\xi_l^{-j})^{-1/2}\\
&=\bigg[\sum_{k=0}^\infty\begin{pmatrix}-\frac{1}{2}\\ k\end{pmatrix}(-r\xi\xi_l^j)^k\bigg]\cdot
 \bigg[\sum_{h=0}^\infty\begin{pmatrix}-\frac{1}{2}\\ h\end{pmatrix}(-r\xi^{-1}\xi_l^{-j})^h\bigg]\\
&=\sum_{k,h=0}^\infty\begin{pmatrix}-\frac{1}{2}\\ k\end{pmatrix}\begin{pmatrix}-\frac{1}{2}\\ h\end{pmatrix}(-r)^{k+h}(\xi\xi_l^j)^{k-h}\\
&=\sum_{n=-\infty}^\infty\begin{pmatrix}(-1)^n\sum_{\substack{k-h=n \\ k,h\geq0}}\begin{pmatrix}-\frac{1}{2}\\ k\end{pmatrix}\begin{pmatrix}-\frac{1}{2}\\ h\end{pmatrix}r^{k+h}\end{pmatrix}(\xi\xi_l^j)^n\\
&\triangleq \sum_{n=-\infty}^\infty b_n (\xi\xi_l^j)^n.
\end{aligned}$$

since $\Gamma(\frac{1}{2})=\sqrt{\pi}$, and the Beta function
$B(x,y)=\int_0^1t^{x-1}(1-t)^{y-1}dt=\frac{\Gamma(x)\Gamma(y)}{\Gamma(x+y)}$ (\cite{handbook}), we have
$$\begin{aligned}
\begin{pmatrix}-\frac{1}{2}\\ k\end{pmatrix}
&=\frac{(-\frac{1}{2})(-\frac{1}{2}-1)\dots(-\frac{1}{2}-k+1)}{k!}\\
&=(-1)^k\frac{\Gamma(k+\frac{1}{2})}{\Gamma(k+1)\Gamma(\frac{1}{2})}\\
&=\frac{(-1)^k}{\pi}\frac{\Gamma(k+\frac{1}{2})\Gamma(\frac{1}{2})}{\Gamma(k+1)}\\
&=\frac{(-1)^k}{\pi}\int_0^1\frac{t^k}{(1-t)^{1/2}t^{1/2}}dt,
\end{aligned}$$

and
$$\begin{aligned}
b_n
&=(-1)^n\sum_{\substack{k-h=n \\ k,h\geq0}}\frac{(-1)^k}{\pi}\int_0^1\frac{t^k}{(1-t)^{1/2}t^{1/2}}dt\begin{pmatrix}-\frac{1}{2}\\ h\end{pmatrix}r^{k+h}\\
&=\sum_{h=0}^\infty\frac{(-1)^h}{\pi}\int_0^1\frac{t^{h+n}}{(1-t)^{1/2}t^{1/2}}\begin{pmatrix}-\frac{1}{2}\\ h\end{pmatrix}r^{2h+n}dt\\
&=\frac{1}{\pi}\int_0^1\frac{(tr)^{n}}{(1-t)^{1/2}t^{1/2}}(1-tr^2)^{-1/2}dt,\\
\end{aligned}$$
since $b_n=b_{-n}$ and $\sum_{j=1}^l\xi_l^{jn}=\left\{\begin{matrix}0 & \ if\ n\not\equiv0\mod l,\\ l  & \ if \ n\equiv0\mod l,
\end{matrix}\right.$
so we have
$$\begin{aligned}
\frac{1}{l}\sum_{j=1}^l|1-r\xi\xi_l^j|
&=\frac{1}{l}\sum_{j=1}^l\sum_{n=-\infty}^\infty b_n (\xi\xi_l^j)^n=\sum_{n=0}^\infty b_{ln} \xi^{ln} +\sum_{n=1}^\infty b_{ln} \xi^{-ln}\\
&=\sum_{n=0}^\infty  \frac{1}{\pi}\int_0^1\frac{(tr)^{ln}}{(1-t)^{1/2}t^{1/2}}(1-tr^2)^{-1/2}\xi^{ln}dt\\
&\ \ \ +\sum_{n=1}^\infty  \frac{1}{\pi}\int_0^1\frac{(tr)^{ln}}{(1-t)^{1/2}t^{1/2}}(1-tr^2)^{-1/2}\xi^{-ln}dt\\
&=\frac{1}{\pi}\int_0^1\frac{(1-tr^2)^{-1/2}}{(1-t)^{1/2}t^{1/2}}
[\sum_{n=0}^\infty(tr\xi)^{ln}+\sum_{n=1}^\infty(tr\xi^{-1})^{ln}]dt\\
&=\frac{1}{\pi}\int_0^1\frac{(1-t)^{-1/2}t^{-1/2}}{(1-tr^2)^{1/2}}
\frac{1-(tr)^{2l}}{|1-(tr\xi)^l|^{2}}dt.
\end{aligned}$$
Thus the integral representation (\ref{potential3}) holds.

\end{proof}


\begin{lem}\label{lem1}
The potential $U$ of the test loop $\tilde{u}$
$$U(\tilde{u})\leq\max\{U(0,\frac{\pi}{2l}),U(\frac{\pi}{2l},0)\}.$$
\end{lem}
\begin{proof}
By (\ref{potential3})
$$U(r,\theta)=\frac{n(1+r)}{4a\sqrt{r}}
\Big[C_l+\frac{l\sqrt{r}}{\pi}\int_0^1\frac{(1-t)^{-1/2}t^{-1/2}}{(1-tr^2)^{1/2}}
\frac{1-(tr)^{2l}}{1+(tr)^{2l}-2(tr)^l\cos(2l\theta)}dt\Big],$$
which implies $\frac{\partial U}{\partial \theta}<0$ for $0<\theta<\frac{\pi}{2l}$, $0<\varphi<\frac{\pi}{2}$. Thus we get for given $0<\varphi<\frac{\pi}{2}$,
$$\sup_{\theta\in[0,\frac{\pi}{2l}]}U(\varphi,\theta)=U(\varphi,0).$$
From (\ref{potential1}), we have
$$\begin{aligned}
\frac{\partial U}{\partial \varphi}
&=\frac{n}{2a}\frac{\sin\varphi}{\cos^2\varphi}
[C_l-\sum_{j=1}^l\frac{\cos^2(j\pi/l-\theta)}{(\sin^2(j\pi/l-\theta)+\tan^2\varphi)^{3/2}}]\\
&=\frac{n}{2a}\frac{\sin\varphi}{\cos^2\varphi}[f_\theta(\varphi)].
\end{aligned}$$
Obviously, $f_\theta(\varphi)$ has the same sign with $\frac{\partial U}{\partial \varphi}$ and $f_\theta(\varphi)$ is  monotonic increasing. Since $\lim_{\varphi\rightarrow\frac{\pi}{2}}f_\theta(\varphi)=C_l>0$, we see that for given $\theta\in(0,\frac{\pi}{2l})$, $\bar{\varphi}\in[0,\frac{\pi}{2}),$
$$\sup_{\varphi\in[0,\bar{\varphi}]}U(\varphi,\theta)=\max\{U(0,\theta),U(\bar{\varphi},\theta)\},$$
and the conclusion follows.
\end{proof}

\begin{lem}\label{cnestimate}\label{lem2}
For all $ n\in\mathbb{N}^+$, $C_n<\frac{n}{\pi}(\log{n}+\gamma)$, then
$U(0,\frac{\pi}{n})=\frac{n}{2a}C_n<\frac{n^2}{2a\pi}(\log{n}+\gamma)$.
\end{lem}
\begin{proof}
Let $\gamma$ be the Euler-Mascheroni constant, that is
$$\gamma=\lim_{n\rightarrow\infty}(1+\frac{1}{2}+\dots+\frac{1}{n}-\log n).$$
Let $a_n=1+\frac{1}{2}+\dots+\frac{1}{n}-\log n$ and $b_n=1+\frac{1}{2}+\dots+\frac{1}{n-1}-\log n$, then
$$a_{n+1}-a_n=\frac{1}{n+1}-\log\frac{n+1}{n}<0,$$
$$b_{n+1}-b_n=\frac{1}{n}-\log\frac{n+1}{n}>0,$$
which implies
$$b_n<\gamma<a_n.$$
Thus we get
$$\sum_{j=1}^{n-1}\frac{1}{j}<\log n+\gamma.$$
We notice that, for $0<x<1$, we have
 $$\pi<\frac{\sin\pi x}{x(1-x)}\leq4,$$
which implies
\begin{equation*}\label{cscestimate}
\csc \frac{j}{n}\pi<\frac{1}{\pi}\frac{1}{\frac{j}{n}(1-\frac{j}{n})}
=\frac{n}{\pi}(\frac{1}{j}+\frac{1}{n-j}),
\end{equation*}
then
$$C_n=\frac{1}{2}\sum_{j=1}^{n-1}\csc \frac{j}{n}\pi<\frac{n}{\pi}\sum_{j=1}^{n-1}\frac{1}{j}
<\frac{n}{\pi}(\log n+\gamma),$$
and the conclusion immediately follows from (\ref{potential1}).



\end{proof}

\begin{remark}
There is an asymptotic expansion of $C_n$ in Lemma 1 of \cite{moeckel1995bifurcation} for large  $n$:
$$C_n\sim \frac{n}{\pi}(\gamma+\log\frac{2n}{\pi})+2\sum_{k\geq1}
\frac{(-1)^k(2^{2k-1}-1)B_{2k}^2\pi^{2k-1}}{(2k)(2k)!}\frac{1}{n^{2k-1}},$$
where  $B_{2k}$ stands for the Bernoulli numbers.
\end{remark}
\begin{lem} \label{lem3}
For $n=2l\geq8$, we have $U(\frac{\pi}{2l},0)<\frac{n^2}{2a\pi}(\log{n}+\gamma)$.
\end{lem}

%
%

\begin{proof}
From (\ref{potential1}), we have
$$\begin{aligned}
U(\frac{\pi}{n},0)
&=\frac{n}{4a\cos\frac{\pi}{n}}\{2C_l+\sum_{j=1}^l
[\sin^2(\frac{j\pi}{l})+\tan^2\frac{\pi}{n}]^{-\frac{1}{2}}\}\\
&=\frac{n}{4a\cos\frac{\pi}{n}}\{2C_l+\sum_{j=1}^{l-1}
[\sin^2(\frac{j\pi}{l})+\tan^2\frac{\pi}{n}]^{-\frac{1}{2}}+\cot\frac{\pi}{n}\} \\
&<\frac{n}{4a\cos\frac{\pi}{n}}(4C_l+\cot\frac{\pi}{n}) =
\frac{n}{4a}(\frac{4C_l}{\cos\frac{\pi}{n}}+\frac{1}{\sin\frac{\pi}{n}})\\
&<\frac{n}{4a}(\frac{4C_l}{1-\frac{\pi^2}{2n^2}}+\frac{1}{\frac{\pi}{n}-\frac{\pi^3}{6n^3}}),
\end{aligned}$$

Thus by Lemma \ref{cnestimate} we have
$$\begin{aligned}
U(\frac{\pi}{n},0)
&<\frac{n^2}{2a\pi}[\frac{\log \frac{n}{2}+\gamma}{1-\frac{\pi^2}{2n^2}}+\frac{1}{2}\frac{1}{1-\frac{\pi^2}{6n^2}}]\\
&=\frac{n^2}{2a\pi}[(1+\frac{\frac{\pi^2}{2n^2}}{1-\frac{\pi^2}{2n^2}})(\log \frac{n}{2}+\gamma)+\frac{1}{2}\frac{1}{1-\frac{\pi^2}{6n^2}}]\\
&=\frac{n^2}{2a\pi}[\log \frac{n}{2}+\gamma+\frac{\pi^2}{2n^2-\pi^2}(\log \frac{n}{2}+\gamma)+\frac{1}{2}\frac{1}{1-\frac{\pi^2}{6n^2}}],
\end{aligned}$$

since $n\geq8$ and $\gamma\approx0.57721566490153286$,
$$\begin{aligned}
U(\frac{\pi}{n},0)
&<\frac{n^2}{2a\pi}[\log n+\gamma-\log2+\frac{\pi^2}{128-\pi^2}(\log4+\gamma)+
\frac{1}{2}\frac{1}{1-\frac{\pi^2}{6\times64}}]\\
&<\frac{n^2}{2a\pi}(\log n+\gamma).
\end{aligned}$$
\end{proof}
Then the estimate (\ref{potential estimate}) holds from Lemma \ref{lem1},  \ref{lem2} and \ref{lem3}.
\begin{remark}
To finish the proof of Remark \ref{8}, we only need to prove $\mathcal{A}(\tilde{u})<\mathcal{B}$ for $n=8$ and $s=h=2$. Throughout this remark, we keep in mind that $n=8$ and  $\mathcal{B}=42\times2^{1/3}\pi^{2/3}T^{1/3}$. It is  direct computation that $$U(\frac{\pi}{8},0)<U(0,\frac{\pi}{8})=\frac{4}{a}C_8=\frac{2}{a}\sum_{j=1}^7\csc \frac{j\pi}{8}<\frac{24}{a},$$
then by Lemma \ref{cnestimate}, the potential of the the test loop $U(\tilde{u})<\frac{24}{a}$. Similar to the proof of Proposition \ref{upperbound}, we have
$$\mathcal{A}(\tilde{u})=\int_0^T[K+U ]dt<\int_0^T[\frac{36a^2\pi^2}{T^2}+\frac{24}{a}] dt,$$
we choose $a=3^{-1/3}\pi^{-2/3}T^{2/3}$, then
$\mathcal{A}(\tilde{u})<36\times3^{1/3}\pi^{2/3}T^{1/3}<\mathcal{B}.$

\end{remark}

\subsection{Partial Collision}\label{secpc}
In this section we prove the following theorem and the idea  is mainly from Fusco et al.\cite{fusco2011platonic}.
\begin{thrm}
A minimizer $u_*$ of $\mathcal{A}|_{\mathcal{K}_s}$ is free of partial collisions.
\end{thrm}

The proof is by contradiction. Suppose $u_*$ is a minimizer of $\mathcal{A}|_{\mathcal{K}_s}$ with partial collisions at time $t_c\in[0,T]$, then we prove that there is  some local deformation which has lower action and so the minimizer is collision free. With the following lemma, we can only consider $u_*$ with partial collisions at time $t_c$ and without  collisions in the time interval $[t_c-\tau,t_c)\bigcup(t_c,t_c+\tau]$.

\begin{lem}
Suppose $u_*$ is a minimizer of $\mathcal{A}|_{\mathcal{K}_s}$ with partial collisions at time $t_c$, then the collision is isolated.
\end{lem}
\begin{proof}
A collision is called isolated at $t_c$ means it is an isolated point in the set of collision times and the lemma is just the Corollary 5.12 in \cite{ferrario2004existence}, so we omit the proof.
\end{proof}

As mentioned in Remark \ref{foundamental}, it is enough to consider $t_c\in\mathbb{I}=[0,\frac{T}{2h}]$. Moreover we see that $u_*\in\Lambda_s$ has the symmetry (\ref{ca}), it is obvious that the collision must happen at $\Gamma\setminus0$. So in the following, we  only discuss partial collisions in two situations: 1. colliding at $\xi_3$-axis and it is two regular $l$-polygonal collisions (Figure 3); 2. colliding in the $\xi_1\xi_2$-plane and it is $l$ binary collisions (Figure   4).

\begin{figure}
\tdplotsetmaincoords{60}{110}
\begin{tikzpicture}[scale=3,tdplot_main_coords]
\draw[thick,<-] (-1.5,0,0) node[left]{$\xi_2$} -- (1.5,0,0);
\draw[thick,->] (0,-1.5,0) -- (0,1.5,0) node[anchor=north west]{$\xi_1$};
\draw[thick,->] (0,0,-1) -- (0,0,1) node[anchor=south]{$\xi_3$};

\filldraw [gray] (0,0.1,0.9) circle (0.5pt)
                 (0,-0.1,0.9) circle (0.5pt)
                 (0.0866,-0.05,0.9) circle (0.5pt)
                 (-0.0866,0.05,0.9) circle (0.5pt)
                 (-0.0866,-0.05,0.9) circle (0.5pt)
                 (0.0866,0.05,0.9) circle (0.5pt);
\draw[very thin] (0,0.1,0.9) -- (-0.0866,0.05,0.9)-- (-0.0866,-0.05,0.9)--(0,-0.1,0.9)--(0.0866,-0.05,0.9)--(0.0866,0.05,0.9)--(0,0.1,0.9);

\filldraw [gray] (0,0.1,-0.9) circle (0.5pt)
                 (0,-0.1,-0.9) circle (0.5pt)
                 (0.0866,-0.05,-0.9) circle (0.5pt)
                 (-0.0866,0.05,-0.9) circle (0.5pt)
                 (-0.0866,-0.05,-0.9) circle (0.5pt)
                 (0.0866,0.05,-0.9) circle (0.5pt);
\draw[very thin] (0,0.1,-0.9) -- (-0.0866,0.05,-0.9)-- (-0.0866,-0.05,-0.9)--(0,-0.1,-0.9)--(0.0866,-0.05,-0.9)--(0.0866,0.05,-0.9)--(0,0.1,-0.9);
\end{tikzpicture}
\caption{}
\end{figure}


\begin{figure}
\tdplotsetmaincoords{60}{110}
\begin{tikzpicture}[scale=2,tdplot_main_coords]
\draw[thick,<-] (-1.5,0,0) node[left]{$\xi_2$} -- (1.5,0,0);
\draw[thick,->] (0,-1.5,0) -- (0,1.5,0) node[anchor=north west]{$\xi_1$};
\draw[thick,->] (0,0,-1) -- (0,0,1) node[anchor=south]{$\xi_3$};

\filldraw [gray] (0,1,0.1) circle (0.5pt)
                 (0,-1,0.1) circle (0.5pt)
                 (0.866,-0.5,0.1) circle (0.5pt)
                 (-0.866,0.5,0.1) circle (0.5pt)
                 (-0.866,-0.5,0.1) circle (0.5pt)
                 (0.866,0.5,0.1) circle (0.5pt);
\draw[very thin] (0,1,0.1) -- (-0.866,0.5,0.1)-- (-0.866,-0.5,0.1)--(0,-1,0.1)--(0.866,-0.5,0.1)--(0.866,0.5,0.1)--(0,1,0.1);

\filldraw [gray] (0,1,-0.1) circle (0.5pt)
                 (0,-1,-0.1) circle (0.5pt)
                 (0.866,-0.5,-0.1) circle (0.5pt)
                 (-0.866,0.5,-0.1) circle (0.5pt)
                 (-0.866,-0.5,-0.1) circle (0.5pt)
                 (0.866,0.5,-0.1) circle (0.5pt);
\draw[very thin] (0,1,-0.1) -- (-0.866,0.5,-0.1)-- (-0.866,-0.5,-0.1)--(0,-1,-0.1)--(0.866,-0.5,-0.1)--(0.866,0.5,-0.1)--(0,1,-0.1);

\draw[very thin,->] (0,1,-0.1) -- (0,1,0);\draw[->] (0,1,0.1)--(0,1,0);
\draw[very thin,->] (0,-1,-0.1) -- (0,-1,0);\draw[->] (0,-1,0.1)--(0,-1,0);
\draw[very thin,->] (-0.866,0.5,-0.1) -- (-0.866,0.5,0);\draw[->] (-0.866,0.5,0.1)--(-0.866,0.5,0);
\draw[very thin,->] (0.866,-0.5,-0.1) -- (0.866,-0.5,0);\draw[->] (0.866,-0.5,0.1)--(0.866,-0.5,0);
\draw[very thin,->] (0.866,0.5,-0.1) -- (0.866,0.5,0);\draw[->] (0.866,0.5,0.1)--(0.866,0.5,0);
\draw[very thin,->] (-0.866,-0.5,-0.1) -- (-0.866,-0.5,0);\draw[->] (-0.866,-0.5,0.1)--(-0.866,-0.5,0);
\draw[very thin,->] (-0.866,0.5,-0.1) -- (-0.866,0.5,0);\draw[->] (-0.866,0.5,0.1)--(-0.866,0.5,0);
\end{tikzpicture}
\caption{}
\end{figure}
We notice that $u_*\in \mathcal{K}_s$ implies that there are both symmetry and topological constraints on $u_0(t)$ at time $t=0,\frac{T}{2h}$.
 For example, for $\epsilon$ small, $\forall t\in(0,\epsilon)$, $u_0(t)=\tilde{R}_lu_0(-t)$ and $u_0(0)\in P_0^-$.
 This doesn't allow general perturbations since $u_0(0)$ is not in the whole plane $P_0$, for this reason we can not use Marchal's idea of averaging the action over sphere or its extension in \cite{ferrario2004existence} for averaging over suitable circles (rotating circle property). But for $t\in(0,\frac{T}{2h})$, such technique works, so   we only put our focus on the partial collision at time $t=0$ (and the case $t=\frac{T}{2h}$ is similar).

Let $\mathbf{k}\subset\mathbf{n}$ be the colliding cluster,
and $q(t)=(q_j(t))_{j\in\mathbf{k}}$ are the trajectories of the colliding cluster.
Denote the partial Lagrangian $\mathcal{L}_\mathbf{k}\big(q(t)\big)=K_\mathbf{k}+U_\mathbf{k}$, where   $K_\mathbf{k}=\sum_{j\in\mathbf{k}}\frac{1}{2}m_j|\dot{q}_j|^2$ is the partial kinetic and $U_\mathbf{k}=\sum_{i,j\in\mathbf{k},i<j}\frac{m_im_j}{|x_i-x_j| }$ is the partial potential function .
Thanks to the blow-up technique  (Lemma \ref{bu1}-\ref{bu3}), it is enough to consider a parabolic collision-ejection orbit
$\bar{q}(t)=(\kappa t)^\frac{2}{3}\bar{s}$ instead of $q(t)$, where $\bar{s}$ is a normalized central configurations. 
Also since the ``blow-up" sends all other bodies not concerned by the collision cluster $\mathbf{k}$ to infinity, we can only do deformation inside $\mathbf{k}$ (for more detail, see Section 7 in \cite{ferrario2004existence}).

\begin{lem}\label{bu1}(Proposition 6.25 in \cite{ferrario2004existence})
Let $\mathbf{k}$ be a colliding cluster, $I_k=\sum_{j\in\mathbf{k}}m_iq_i^2$: Then there is $\kappa>0$ such that the following asymptotic estimate hold:
$$\begin{aligned}
I_\mathbf{k}&\sim (\kappa t)^\frac{4}{3}\\
\dot{I}_\mathbf{k}&\sim \frac{4}{3}\kappa(\kappa t)^\frac{1}{3}\\
\ddot{I}_\mathbf{k}&\sim  \frac{4}{9}\kappa^2(\kappa t)^\frac{-2}{3}.
\end{aligned}$$
and therefore
$$K_\mathbf{k}\sim U_\mathbf{k}\sim \frac{1}{2}\ddot{I}_\mathbf{k}\sim \frac{2}{9}\kappa^2(\kappa t)^\frac{-2 }{3}.$$
\end{lem}

Let $s=I_\mathbf{k}^{-\frac{1}{2}}q$ be the normalized configuration, we have

\begin{lem}\label{bu2}(Proposition 6.32 in \cite{ferrario2004existence})
For every converging sequence $s(t_j)$ of  normalized configuration.  The limit $\bar{s}=\lim_{j\rightarrow\infty}s(t_j)$ is a central configuration.
\end{lem}

We say that $\bar{q}$ is a (right) blow-up of the solution $x(t)$ in $0$, if
$$\bar{q}=(\kappa t)^\frac{2}{3}\bar{s}.$$
For every $\lambda>0$ consider the path $x^\lambda$ defined by
$$x^\lambda(t)=\lambda^{-\frac{2}{3}}x(\lambda t)$$
for every $t\in [0,\lambda^{-1}\epsilon]$, we have

%
%

\begin{lem}\label{bu3}(Proposition 7.9 in \cite{ferrario2004existence})
Let $x(t$) be a solution in $(0, \epsilon)$, with an isolated collision in t = 0. Let $\mathbf{k}\subset\mathbf{n}$ be a colliding cluster and $\bar{q}$ a (right) blow-up of $x(t)$ with respect to $\mathbf{k}$ in $0$. Let $\tau\in(0,\epsilon)$ and let $\varphi$ be a variation of the particles in $\mathbf{k}$ which is $C^1$ in a neighborhood of $\tau$, defined and centered for every
$t\in[0, \tau]$. Then there exists a sequence $\psi_n\in H^1([0,\tau],\mathbb{R}^3)^k$ of centered paths converging uniformly to $0$ in $[0,\tau]$ and with support in $[0,\tau]$, has the following property:
$$\lim_{n\rightarrow\infty}\int_0^\tau[\mathcal{L}(x^{\lambda_n}+\varphi+\psi_n)-\mathcal{L}(x^{\lambda_n})]dt
=\int_0^\tau[\mathcal{L}(\bar{q}+\varphi)-\mathcal{L}(\bar{q})]dt.$$
\end{lem}

Furthermore, our deformation is based on the solution of the two-body problem(\cite{twobody}\cite{Chenciner2003ICM}).
Suppose $\omega:\mathbb{R}\rightarrow \mathbb{R}^2$ is the ejection-collision solution of
 $$\ddot{\omega}=-\frac{\beta\omega}{|\omega|^3}, \quad \beta>0,$$
i.e. $\omega(\pm t)=(\frac{9}{2}\beta)^{1/3}t^{2/3}n^\pm ,\ t\geq 0$, where $n^\pm=\lim_{t\rightarrow 0^+}\frac{\omega(\pm t)-\omega(0)}{|\omega(\pm t)-\omega(0)|}.$
Let $\theta_d=\arccos(n^+\cdot n^-)$ and $\theta_i=2\pi-\theta_d$. If $\theta_d\in(0,\pi]$, given $\epsilon>0$, there are exactly two Keplerian arcs $\omega_d:[-\epsilon,\epsilon]\rightarrow \mathbb{R}^2$ and $\omega_i:[-\epsilon,\epsilon]\rightarrow \mathbb{R}^2$ which connect $\omega(-\tau)$ to $\omega(\tau)$ in the time interval $[-\tau,\tau]$ and satisfy
$$\left\{\begin{aligned}
 &\omega_d((-\epsilon,\epsilon))\subset\{a_1n^-+a_2n^+,a_i>0\},\\
 &\omega_i((-\epsilon,\epsilon))\subset span\{n^-,n^+\}\setminus\{a_1n^-+a_2n^+,a_i>0\}.
\end{aligned}\right.$$
The arcs $\omega_d$ and $\omega_i$ are called the direct and indirect Keplerian arc. When $\theta_d=0$, the indirect arc does not exist and $\omega_d((-\epsilon,\epsilon))\subset\{an^+,a>0\}$.
\begin{lem}(\cite{twobody}\cite{fusco2011platonic})\label{keplerarc}
The following inequalities hold:
\begin{enumerate}
\item $A(\omega_d)<A(\omega|_{[-\epsilon,\epsilon]})$, $\forall n^\pm$,
\item $A(\omega_i)<A(\omega|_{[-\epsilon,\epsilon]})$, $\forall n^\pm$ such that $n^+\cdot n^-<1$.
\end{enumerate}
where
$$A(\omega)=\int_{-\epsilon}^\epsilon(\frac{|\dot{\omega}|^2}{2}+\frac{\beta}{|\omega| })$$
\end{lem}

\begin{proof}
There is a detailed proof in Proposition 5.7 of \cite{fusco2011platonic}.
\end{proof}

\subsubsection{colliding in $\xi_3$-axis at time $t=0$}\label{partial1}
In this situation, $u_0(0)\in \xi_3\bigcap P_0^-$ and since it is a partial collision, we must have  $u_0(0)\cdot e_3<0$.  We can do perturbation in the whole plane $\xi_3=u_0(0)\cdot e_3$.\\

\begin{figure}

\begin{tikzpicture}[scale=2]
\draw [->](-1.5,0) -- (1.5,0) node[anchor=north] {$\xi_2$};
\draw [->](0,-1.5) -- (0,1.5)node[right] {$\xi_3$};
\draw [->](0,0) -- (0.15,0.4);\draw (0.15,0.4) -- (0.3,0.8);
\draw [->](-0.3,0.8) -- (-0.15,0.4);\draw (-0.15,0.4) -- (0,0);
\filldraw [gray] (0,0) circle (0.5pt)
(0.3,0.8) circle (0.5pt)
(-0.3,0.8) circle (0.5pt);
\draw [very thin,->](-0.3,0.8) parabola bend (0,-0.2) (0.3,0.8);
\draw (0,-0.2)  node[right]  {$\omega(t)$};
\draw (0,0.16) .. controls (0.03,0.18)  .. (0.06,0.16);
\draw (0.045,0.28) node[] {$\theta$};
\draw (0,0.16) .. controls (-0.03,0.18)  .. (-0.06,0.16);
\draw (-0.045,0.28) node[] {$\theta$};
\end{tikzpicture}
\caption{}
\end{figure}

The collision set at $t=0$ is $q=\{q_0,q_1,\dots,q_{l-1}\}$ and it is a collision of regular $l$-polygon, so the partial functional
$$\int_0^\tau\mathcal{L}_\mathbf{k}(q)dt=l\int_0^\tau\big(\frac{1}{2}|\dot{q}_0|^2+
\frac{\alpha}{|q_0|}\big)dt,$$
where $\alpha=\frac{1}{4}\sum_{j=1}^{l-1}\csc(j\pi/l)$.
Let $\bar{q}$ be the right blow-up of $q$, then $\bar{q}_0$ is the collision-ejection in the plane $\xi_3=q_0(0)\cdot e_3$.
By Lemma \ref{keplerarc}, there exists a Keplerian direct arc $\omega(t)$ satisfying
$\int_0^\tau\big(\frac{1}{2}|\dot{\omega}_0|^2+\frac{\alpha}{|\omega_0|}\big)dt<
\int_0^\tau\big(\frac{1}{2}|\dot{\bar{q}}_0|^2+\frac{\alpha}{|\bar{q}_0|}\big)dt$ and if we let $\tilde{q}_0(t)=(\omega(t),\bar{q}_0(t)\cdot e_3)$, we have $\tilde{q}_0(\tau)=\bar{q}_0(\tau)$. By the symmetric condition (\ref{cb}), we must have $\tilde{q}_0(0)\in P_0^-\setminus\xi_3$.
Consequently, the perturbation $\tilde{q}=\{\tilde{q}_0,R_1\tilde{q}_0,\dots,R_{l-1}\tilde{q}_0\}$  satisfies $\int_0^\tau\mathcal{L}_\mathbf{k}(\tilde{q})dt<\int_0^\tau\mathcal{L}_\mathbf{k}(\bar{q})dt$.
By Lemma \ref{bu1}-\ref{bu3},  there is no partial collisions  in $\xi_3$-axis.

\subsubsection{colliding in the $\xi_1\xi_2$-plane at time $t=0$}


The partial collision is the union of $l$ binary collisions. This situation is   more complicated since we have the topological constraint $u_0(0)\cdot e_3<0$ which implies that $\omega(t)$ should be the indirect arc when $0\leq \theta\leq\frac{\pi}{2}$(Figure 5).
But the indirect arc does not exist when $\theta=0$,
so similar to the proof of Section \ref{partial1}, there is no  partial collisions unless  $n^\pm=\lim_{t\rightarrow 0^\pm}\frac{u_0(t)-u_0(0)}{|u_0(t)-u_0(0)|}=e_3$ (see Figure 4).
\begin{remark}
Such collision described in the above is called the collision of type $(\rightrightarrows)$ in Definition 5.1 of \cite{fusco2011platonic}. Let $u_*\in\overline{\mathcal{K}}$ be a minimizer of the action $\mathcal{A}|_\mathcal{K}$ and assume that
$u_*$ has a partial collision at time $t_c$. Let $r$ be the axis on which the collision of the generating particle takes place and $n^+$,$n^-$ be the unit vectors associated to the collision. Then we say that the collision is of type $(\rightrightarrows)$ if\\
(1) $n^+=n^-$,\\
(2) The plane generated by $r$, $n = n^\pm$ is fixed by some reflection $\tilde{R}\in\tilde{G}$.
\end{remark}

\begin{lem}\label{123}(Corollary 5.1 in \cite{fusco2011platonic})\\
Let $\omega:(0,\bar{t})\rightarrow\mathbb{R}^3$(or $(-\bar{t},0)\rightarrow\mathbb{R}^3$) be a maximal ejection
(collision) solution to the equation
\begin{equation} \label{eq5.7}
\ddot{\omega}=a\frac{(R_\pi-I)\omega}{|(R_\pi-I)\omega|^3}+V_1(\omega),
\end{equation}
and let
$n=\lim_{t\rightarrow 0^+}\frac{\omega(t)-\omega(0)}{|\omega(t)-\omega(0)|}=\lim_{t\rightarrow 0^-}\frac{\omega(t)-\omega(0)}{|\omega(t)-\omega(0)|}$ be
the unit vector orthogonal to $r$. Assume that the plane $\pi_{r,n}$, generated by $r,n$, is fixed by some reflection $\tilde{R}$. Then
$$\omega(t)\in\pi_{r,n},\ \forall t\in (0,\bar{t}),(or\ \forall t\in (-\bar{t},0)).$$
In (\ref{eq5.7}) $R_\pi$ denotes the rotation of $\pi$ around the axis $r$, $V_1$ is a smooth function defined in an open set $\Omega\subset\mathbb{R}^3$   containing $r\setminus{0}$ and $a\in \mathbb{R}$. Moreover $V_1$ satisfies the symmetry condition $V_1(\tilde{R}\omega)=\tilde{R}V_1(\omega)$, where $\tilde{R}$ is a reflection such that $\tilde{R}r=r$.
\end{lem}

By the above arguments, the collision of type  $(\rightrightarrows)$  is the only collision at time $t=0$.
Since there is no collision in $(0,\frac{T}{2h})$,  we can apply  this lemma to our context and let $\bar{t}=\frac{T}{2h}$, then  $u_0(t)\in P_0,\ \forall t\in (0,\frac{T}{2h})$, which is a contradiction with $u_0(\frac{T}{2h})\in P_s$ in our assumption, thus  there is no partial collision at time $t=0$. And the case for $t=\frac{T}{2h}$ is similar, so we have finished the proof.
\begin{remark}
Lemma \ref{123} is actually right in the view of physics.  If there is some plane $\pi$ such that   the particle force $\frac{\partial V}{\partial \omega}\in \pi$ when  $\omega\in\pi$. Then $\omega(0),\dot{\omega}(0) \in\pi$ must imply that $\omega(t)\in\pi$ for all $t>0$.
\end{remark}


\newpage
\begin{thebibliography}{10}

\bibitem{handbook}
M. Abramowitz and I. Stegun,
\newblock {Handbook of mathematical functions with formulas, graphs, and
  mathematical tables}.
\newblock {\em Mathematics of Computation}, 1964.

\bibitem{twobody}
A. Albouy,
\newblock {Lectures on the two body problem. In: Carbral,H., Diacu,F.(eds.)
  Classical and Celestial Mechanics (The Recife Lectures).}
\newblock {\em Princeton University Press, Princeton}, 2002.

\bibitem{chen2008existence}
K.C. Chen,
\newblock Existence and minimizing properties of retrograde orbits to the
  three-body problem with various choices of masses.
\newblock {\em Annals of Mathematics}, 167(2008),325--348.

\bibitem{chenciner2002action}
A. Chenciner,
\newblock Action minimizing solutions of the $n$-body problem: from homology to symmetry.
\newblock {\em Proc. ICM}, 2002.

\bibitem{Chenciner2003ICM}
A. Chenciner,
\newblock {Symmetries and ``simple" solutions of the classical n-body problem.}
\newblock {\em ICMP03}, 2003.

\bibitem{chenciner2000remarkable}
A. Chenciner and R. Montgomery,
\newblock A remarkable periodic solution of the three-body problem in the case
  of equal masses.
\newblock {\em Annals of Mathematics-Second Series}, 152(2000):881--902.


\bibitem{0951-7715-21-6-009}
D.L. Ferrario and A. Portaluri,
\newblock On the dihedral n-body problem.
\newblock {\em Nonlinearity}, 21(2008),1307-1321.

\bibitem{ferrario2004existence}
D.L. Ferrario and S. Terracini,
\newblock On the existence of collisionless equivariant minimizers for the
  classical n-body problem.
\newblock {\em Inventiones mathematicae}, 155(2004),305--362.

\bibitem{fusco2011platonic}
G.~Fusco, G.F.~Gronchi, and P.~Negrini,
\newblock Platonic polyhedra, topological constraints and periodic solutions of
  the classical n-body problem.
\newblock {\em Inventiones mathematicae}, 185(2011),283--332.

\bibitem{gordon1977minimizing}
W.~B. Gordon,
\newblock A minimizing property of Keplerian orbits.
\newblock {\em American Journal of Mathematics}, 99(1977),961--971.

\bibitem{finite-reflection-groups}
L.~C. Grove and C.~T. Benson,
\newblock {Finite Reflection Groups}.
\newblock {\em Springer,Berlin}, 1985.

\bibitem{marchal2002method}
C.~Marchal,
\newblock How the method of minimization of action avoids singularities.
\newblock {\em Celestial Mechanics and Dynamical Astronomy}, 83(2002),325--353.

\bibitem{moeckel1995bifurcation}
R. Moeckel and Carles Sim{\'o},
\newblock Bifurcation of spatial central configurations from planar ones.
\newblock {\em SIAM Journal on Mathematical Analysis}, 26(1995),978--998.

\bibitem{montgomery1998braid}
R. Montgomery,
\newblock The N-body problem, the braid group, and action-minimizing periodic solutions.
\newblock {\em Nonlinearity}, 11(1998),363--376.

\bibitem{palais1979principle}
R.~S. Palais,
\newblock The principle of symmetric criticality.
\newblock {\em Communications in Mathematical Physics}, 69(1979),9--30.

\bibitem{shibayama2011minimizing}
M. Shibayama,
\newblock Minimizing periodic orbits with regularizable collisions in the
  n-body problem.
\newblock {\em Archive for rational mechanics and analysis}, 199(2011),821--841.

\bibitem{struwe2008variational}
M. Struwe,
\newblock {\em Variational methods: applications to nonlinear partial
  differential equations and Hamiltonian systems}, volume~34.
\newblock Springer, 2008.

\bibitem{terracini20072nbody}
S. Terracini and A. Venturelli,
\newblock {Symmetric trajectories for the 2N-body problem with equal masses}.
\newblock {\em Arch. Rational Mech. Anal.},
  184(2007),465--493.

\bibitem{16196083}
A. Venturelli,
\newblock {Une caract$\acute{e}$risation variationnelle des solutions de
  Lagrange du probl$\grave{e}$me plan des trois corps}.
\newblock {\em Comptes Rendus De L Academie Des Sciences Serie I-mathematique},
  332(2001),641--644.

\bibitem{16136657}
S.~Q. Zhang and Q. Zhou,
\newblock {A minimizing property of Lagrangian solution}.
\newblock {\em Acta Mathematica Sinica-English Series}, 17(2001),497--500.

\bibitem{2}
S.~Q. Zhang and Q. Zhou,
\newblock {Variational methods for the choreography solution to the three-body problem}.
\newblock {\em Science in China}, 45(2002),594--597.

\bibitem{1}
S.~Q. Zhang and Q. Zhou,
\newblock {Nonplanar and noncollision periodic solutions for N-body problems}.
\newblock {\em Discrete and Continuous Dynamical Systems}, 10(2004),679--685.




\end{thebibliography}
\end{document}